\documentclass[letterpaper, 10pt, conference]{ieeeconf} 
\IEEEoverridecommandlockouts
\usepackage{graphicx}
\usepackage{cite}
\usepackage{url}
\usepackage{xcolor}
\usepackage{amsmath}
\usepackage{mathtools}
\usepackage{tikz}
\usepackage{verbatim}
\usepackage{subcaption}
\usepackage{xspace}
\usepackage{setspace} 
\usepackage{pgfplots}
\pgfplotsset{compat=1.10}
\usepackage{graphicx}
\usepackage{epstopdf}
\usepackage{amssymb}
\usepackage{relsize}
\usepackage[textwidth=3.8em,textsize=scriptsize,disable]{todonotes}
\usepackage{algorithm}
\definecolor{Gray}{gray}{0.90}
\usepackage{colortbl}

\usepackage{algorithm}
\usepackage{algorithmicx}
\usepackage{algcompatible}
\usepackage{algpseudocode}
\usepackage{color}

\usepackage{blkarray}
\newtheorem{theorem}{\bf{Theorem}}[section]

\newtheorem{prop}[theorem]{\bf{Proposition}}

\newtheorem{observation}[theorem]{\bf{Observation}}

\newtheorem{remark}[theorem]{\bf{Remark}}
\newenvironment{definition}[1][Definition]{\begin{trivlist}
\item[\hskip \labelsep {\bfseries #1}]}{\end{trivlist}}
\def\qed{\hfill\rule[-1pt]{5pt}{5pt}\par\medskip}
\usepackage{amssymb}
\usepackage{url}

\newcommand{\calN}[0]{\mathcal{N}}

\newcommand{\calD}[0]{\mathcal{D}}
\newcommand{\calM}[0]{\mathcal{M}}

\newcommand{\calC}[0]{\mathcal{C}}

\definecolor{blue}{rgb}{0,0,0}
\definecolor{amber}{rgb}{1.0, 0.75, 0.0}

\begin{document}
\title{Edge Augmentation with Controllability Constraints in\\ Directed Laplacian Networks} 
\author{Waseem Abbas, Mudassir Shabbir, Yasin~Yaz{\i}c{\i}o\u{g}lu and Xenofon Koutsoukos
\thanks{W. Abbas M~Shabbir and X. Koutsoukos are with the Electrical Engineering and Computer Science Department at Vanderbilt University, Nashville, TN, USA (Emails: \{\texttt{waseem.abbas,mudassir.shabbir, xenofon.koutsoukos\}@vanderbilt.edu)}. Y.~Yaz{\i}c{\i}o\u{g}lu is with the Department of Electrical and Computer Engineering at the University of Minnesota, Minneapolis, MN, USA (Email: \texttt{ayasin@umn.edu}).}
 }

\maketitle
\begin{abstract}
 In this paper, we study the maximum edge augmentation problem in directed Laplacian networks to improve their robustness while preserving lower bounds on their strong structural controllability (SSC). Since adding edges could adversely impact network controllability, the main objective is to maximally densify a given network by selectively adding missing edges while ensuring that SSC of the network does not deteriorate beyond certain levels specified by the SSC bounds. We consider two widely used bounds: first is based on the notion of zero forcing (ZF), and the second relies on the distances between nodes in a graph. We provide an edge augmentation algorithm that adds the maximum number of edges in a graph while preserving the ZF-based SSC bound, and also derive a closed-form expression for the exact number of edges added to the graph. Then, we examine the edge augmentation problem while preserving the distance-based bound and 
 present a randomized algorithm that guarantees an $\alpha$--approximate solution with high probability. Finally, we numerically evaluate and compare these edge augmentation solutions. 
\end{abstract}

\begin{keywords}
Edge augmentation, structural controllability, zero forcing, graph distances.
\end{keywords}

\section{Introduction}
In a networked multi-agent system, a frequent approach to \textcolor{blue}{improve} network connectivity is to systematically increase interconnections between agents. On the one hand, edge augmentation is useful for improving network connectivity, robustness to link failures, and resilience to malicious intrusions, but on the other hand, \textcolor{blue}{adding edges could adversely impact network controllability \cite{abbas2020trade,pasqualetti2018fragility,9086120,lindmark2020investigating}}. 
In this paper, we study the problem of maximum edge augmentation in a directed network of agents with Laplacian dynamics while preserving the controllability specification. 
We consider the network's strong structural controllability (SSC), which depends (apart from the set of input nodes) only on the structure of the underlying graph defined by the edge set of the graph. 
To measure how much of the network is strong structurally controllable with a given set of leader (input) nodes, the concept of the dimension of strong structurally controllable subspace (SSCS) is typically used (Section \ref{subsection:SSCS}). The exact computation of the dimension of SSCS is a hard task, so various graph-theoretic bounds on the dimension of SSCS have been proposed in the literature. We utilize two widely used bounds that are based on the ideas of \emph{Zero Forcing (ZF)} \cite{monshizadeh2014zero,monshizadeh2015strong,trefois2015zero,mousavi2018structural} and \emph{distances between nodes} in graphs \cite{yazicioglu2016graph,YasinCDC2020}. We discuss these bounds in detail in Sections~\ref{subsection:ZFS_bound} and~\ref{subsection:PMI_bound}, respectively. Our main objective is to \emph{add the maximum number of edges in a given directed graph while preserving the lower bound (ZF-based or distance-based) on the dimension of SSCS.} Our contributions are listed below.

\noindent
1) We present an optimal edge augmentation algorithm for adding the maximum number of edges in a directed graph while preserving the ZF-based bound on the dimension of SSCS. We analyze the algorithm and provide a closed-form expression for the number of edges added in the graph. 

 \noindent
2) We also discuss edge augmentation in graphs that preserves the distance-based bound on the dimension of SSCS.  For a given node pair $(u,v)$ in a directed graph, we characterize the optimal solution of the distance preserving edge augmentation problem in which the objective is to add maximum edges in a graph without changing the distance from node $u$ to node $v$.


\noindent
3) We then provide a randomized algorithm that adds maximal edges in a directed graph while preserving the distance-based bound on the dimension of SSCS. We also analyze the approximation ratio of the algorithm.

 \noindent
 4) Finally, we numerically evaluate and compare various edge augmentation solutions.

\textcolor{blue}{We studied the edge augmentation problem while preserving the distance-based bound on SSC in \emph{undirected} networks in \cite{AbbasACC2020}. In this paper, we focus on \emph{directed networks} and consider both the ZF-based bound and the distance-based bound on the dimension of SSCS. It is worth emphasizing that the edge  augmentation  problem  differs  fundamentally  between directed and undirected networks. This is due to the fact that adding an edge between nodes $u$ and $v$ in an undirected network means adding two directed edges (from $u$ to $v$ and $v$ to $u$) simultaneously. However, no such constraint exists when adding edges to directed networks.} This work is also related to \cite{9086120}, which only considers directed networks that are strong structurally controllable and studies the problem of adding edges while retaining their SSC. Our setup is more general because we consider directed networks that are not necessarily strong structurally controllable. \textcolor{blue}{We note that such a setup is very relevant to the notion of \emph{target controllability} in linear networks, where the goal is to control only a subset of agents (targets) instead of the entire network. Since controlling the entire network might not be required in certain applications and could be costly, it is desired to control only target nodes (for instance, \cite{Van2017distance,gao2014target,monshizadeh2015strong}). Moreover, this “partial controllability” specification could improve other network attributes, which could not be improved otherwise due to complete controllability requirements, for instance, robustness to external perturbations (which could be achieved through link additions).} 

\section{Preliminaries and Problem Description}
\label{sec:Prelim}
We consider a network of agents modeled by a directed graph $G = (V, E)$, where $V$ is the set of agents and $E$ is the set of directed edges. An edge from node $u$ to $v$ is denoted by $(u,v)$, and $u$ is the \emph{in-neighbor}, or simply the \emph{neighbor} of $v$. 
We use the terms node and agent interchangeably. The set of all neighbors of $u$, denoted by $\calN_u$, is called the \emph{neighborhood} of $u$. The \emph{distance} from $u$ to $v$ in $G$, denoted by $d_G(u,v)$, is the number of edges in the shortest directed path from $u$ to $v$. Accordingly, $d_G(u,u) = 0$ and $d_G(u,v)=\infty$ if there is no directed path from $u$ to $v$. We may ignore the subscript $G$ when it is clear from the context. The directed graph is \emph{strongly connected}
if there is a path from any node to any other node. The \emph{union} of $G=(V,E)$ and $G'= (V',E')$ is $G\cup G' = (V\cup V', E\cup E')$. The edges in a graph are assigned positive weights by some weighting function:
\begin{equation}
\label{eq:weights}
w: E\rightarrow \mathbb{R}^+,
\end{equation}
where $\mathbb{R}^+$ is the set of positive real numbers. 
\subsection{System Model}
Each agent $u$ in the network has a state $x_u\in\mathbb{R}$, and the overall state of the network is $x\in\mathbb{R}^n$, where $n = |V|$. The network dynamics are given by the following equation:
\begin{equation}
\label{eq:dynamics}
\dot{x} = -L_w x + B u,
\end{equation}
 where $L_w$ is the \emph{weighted Laplacian} matrix of $G$ and defined as
$L_w = (\text{Deg}-A_w)$. Here, $A_w\in\mathbb{R}^{n\times n}$ is the \emph{weighted adjacency} matrix of $G$ whose $uv^{th}$ entry is
\begin{equation}
\label{eq:A_matrix}
[A_w]_{u,v} = \left\{ \begin{array}{cc}
w(u, v) & \text{if } u\ne v, \text{ and } (u,v)\in E\\
0 & \text{otherwise,}
\end{array}\right.
\end{equation}
 and $\text{Deg}\in\mathbb{R}^{n\times n}$ is the \emph{degree} matrix defined as
 \begin{equation}
 [\text{Deg}]_{u,v} = \left\{ \begin{array}{cc}
 \sum_{k=1}^{n}[A_w]_{u,k} & \text{if } u = v,\\
 0 & \text{otherwise.}
 \end{array}
 \right.
 \end{equation}
In \eqref{eq:dynamics}, $B\in\mathbb{R}^{n\times m}$ is the input matrix, where $m$ is the number of inputs, which is equal to the number of leader nodes. If $V_\ell = \{\ell_1,\ell_2,\cdots, \ell_m\}\subset V$ is the set of leader nodes, then 
 \begin{equation}
 \label{eq:B}
 [B]_{u,v} = \left\{\begin{array}{cc}
 1 & \text{if node } u \text{ is also a leader } \ell_v\\
 0 & \text{otherwise.}
 \end{array}\right.  
 \end{equation}
 
\subsection{Strong Structural Controllability (SSC)}
\label{subsection:SSCS}
A state $x'\in\mathbb{R}^n$ is reachable if there is an input $u$ that can drive the network in \eqref{eq:dynamics} from origin (initial state) 
to $x'$ in a finite amount of time. A network $G = (V,E)$ with edge weights defined by $w$ and leader set $V_\ell$ is \emph{completely controllable}, that is every point in $\mathbb{R}^n$ is reachable, if and only if the following \emph{controllability matrix} is full rank.

\begin{equation*}
\small
\Gamma(L_w,V_\ell) = 
\left[
\begin{array}{ccccc}
B & (-L_w) B & (-L_w)^2 B & \cdots & (-L_w)^{n-1} B
\end{array}
\right].
\end{equation*}

The rank of $\Gamma(L_w,V_\ell)$ defines the \emph{dimension} of the controllable subspace consisting of all the reachable states. 
A Laplacian network $G = (V,E)$ with a given set of leader nodes is called \emph{strong structurally controllable (SSC)} if it is completely controllable for any choice of $w$ as in~\eqref{eq:weights}. 
At the same time, the \emph{dimension of strong structurally controllable subspace (SSCS)}, denoted by $\gamma(G,V_\ell)$, is the \emph{minimum rank} of the controllablility matrix $\Gamma(L_w,B)$ over all feasible $w$ (edge weights), i.e., 
\begin{equation}
\gamma(G,V_\ell) = \min\limits_w \;\text{rank}( \Gamma(L_w,V_\ell)).
\end{equation}

\subsection{Problem Formulation}
The main objective in the paper is to identify the maximum number of missing edges in a given network such that the dimension of SSCS of the network is preserved even after adding those edges. Since computing the dimension of SSCS is computationally challenging, we consider its lower bounds, including the zero forcing (ZF) and distance-based bounds (explained in Sections \ref{section:ZFS_strory} and \ref{section:PMI_Story}, respectively). These bounds are tight and have numerous applications~\cite{YasinCDC2020}. 
If $\delta$ is a (ZF-based or distance-based) lower bound  on the dimension of SSCS of the network, then the goal is to maximally densify the graph while maintaining the dimension of SSCS to be at least $\delta$. Formally, we state the problem below.

\textbf{Problem} Let $G = (V,E)$ be a directed network of agents in which $V_\ell\subseteq V$ be a set of leaders and the network dynamics are defined by \eqref{eq:dynamics}. Let the dimension of SSCS of the network be at least $\delta$. Then, find the maximum size edge set $E'$ such that $E\subseteq E'$ and the dimension of SSCS of the network $G' = (V,E')$ with the same set of leaders $V_\ell$ is also at least $\delta$, i.e., $\delta \le \gamma(G',V_\ell)$.
\section{Adding Edges through Zero Forcing Bound}
\label{section:ZFS_strory}
In this section, we present an edge augmentation algorithm that optimally adds edges in a network while preserving the zero forcing-based bound on the dimension of SSCS.
\subsection{Zero Forcing (ZF) Bound for SSC}
\label{subsection:ZFS_bound}
First, we explain the notion of \emph{Zero Forcing} process and its relation to the dimension of SSCS \cite{monshizadeh2014zero,monshizadeh2015strong,trefois2015zero}.

\begin{definition}{\em (Zero Forcing Process)}
Consider a directed graph $G=(V,E)$ such that each node $v\in V$ is initially assigned either a \textit{white} or \textit{black} color. The following coloring defines the zero forcing process:  \emph{if a black colored node $v\in V$ has exactly one white in-neighbor $u$, then change the color of $u$ to black.} We say that $v$ \emph{infected} $u$.\footnote{Since an edge $(u,v)$ indicates that the state of node $u$ is influenced by the state of $v$ in our system model (as in \eqref{eq:dynamics} and \eqref{eq:A_matrix}), we use in-neighbors in the ZF process for consistency. If an edge $(u,v)$ indicates that node $u$ influences node $v$'s state, then out-neighbors should be used in the ZF process, as done in some other works.} 
\end{definition}

\begin{definition}{\em (Derived Set)}
Consider a directed graph $G=(V,E)$ where $V'\subseteq V$ is an initial set of black nodes (also called the \textit{input set}), and apply the zero forcing process until no further color changes are possible. The resulting set of black nodes is the derived set, denoted by $\text{dset}(G,V') \subseteq V$. \textcolor{blue}{For a given set of input nodes, the derived set is unique \cite{work2008zero}}. Moreover, an input set $V'$ is called a \emph{zero forcing set} (ZFS) if $\text{dset}(G,V')=V$. 
\end{definition}



As an example, consider $G$ in Figure~\ref{fig:Derived}, where $V' = \{v_1,v_2\}$ is the set of input nodes. As a result of ZF process, $v_1$ infects $v_3$, $v_2$ infects $v_4$ and the resulting derived set is $\{v_1,v_2,v_3,v_4\}$. 

\begin{figure}[htb]
    \centering
    \begin{subfigure}[b]{0.239\textwidth}
	\centering
	\includegraphics[scale=0.22]{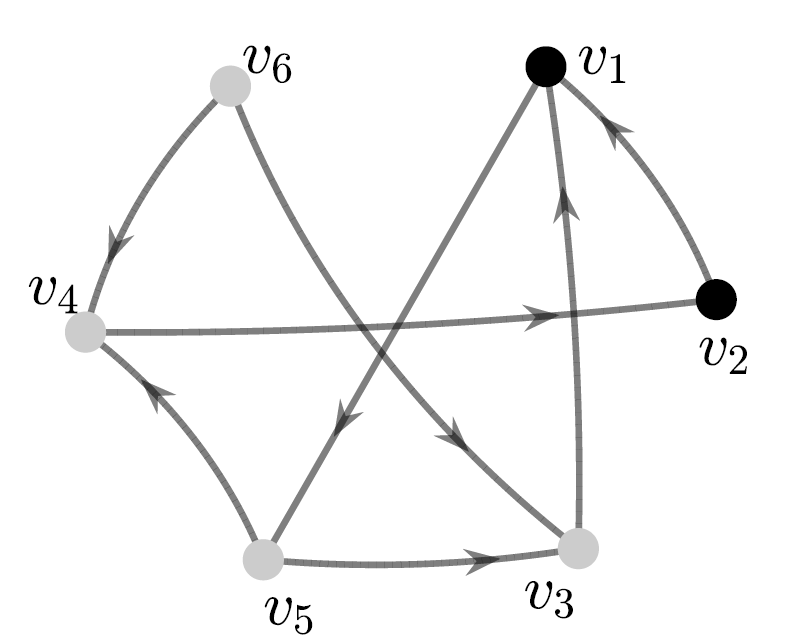}
	\caption{}
	\end{subfigure}
	\begin{subfigure}[b]{0.24\textwidth}
	\centering
\includegraphics[scale=0.33]{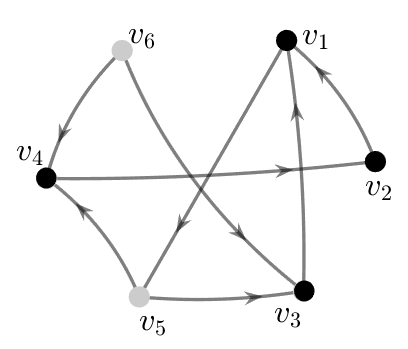}
	\caption{}
	\end{subfigure}
    \caption{An example of zero forcing and derived set.}
    \label{fig:Derived}
\end{figure}

In the context of SSC, the cardinality of the derived set is significant as it provides a lower bound on the dimension of SSCS, as stated in the following result.


\begin{theorem}\cite{monshizadeh2015strong}
For any network $G=(V,E)$ with the leaders $V_\ell \subseteq V$,
\begin{equation}
\label{eq:z_bound}
\zeta(G,V_\ell) \le \gamma(G,V_\ell),
\end{equation}
where $\zeta(G,V_\ell)= |\text{dset}(G,V_\ell)|$ is the size of the derived set corresponding to the input set $V_\ell$.
\end{theorem}
\subsection{Edge Augmentation Algorithm Using ZF}
We provide an algorithm to add edges to a directed graph with a given set of leaders. The algorithm ensures that the derived set of the graph remains the same after adding edges, thus, preserving the ZF-based bound on the dimension of SSCS. The proposed algorithm is a modification of the ZF process. In summary, we look for a BLACK node with a single WHITE in-neighbor, add  \textcolor{blue}{edges incident to the BLACK node that do not change change the size of the derived set}, change the WHITE node's color, and then repeat the same procedure until there is no BLACK node with a single WHITE in-neighbor. When this process concludes, we add any \textcolor{blue}{extra edges that can be added while preserving the derived set}. The algorithm is outlined below. We denote the color of the node $v$ by COLOR($v$).
\begin{algorithm}[!ht]
	\caption{ZF-based Edge Augmentation}
	\label{algo:augmention_zfs}
	{\small
	\begin{algorithmic}
     	\State \textbf{Given} $G= (V,E)$, $V_\ell$ \Comment{$V_\ell$ is a leader set}.
		\State \textbf{Initialize} $G' = (V,E')$, $E'\gets E$
		\State For all $v\in V$, $\text{COLOR}(v)\gets \text{WHITE}$
		\State For all $v\in V_\ell$, $\text{COLOR}(v)\gets \text{BLACK}$
		\While{$\exists u$ a BLACK node with a single WHITE in-neighbor $v$}
		\State $\text{COLOR}(v)\gets \text{BLACK}$
		\State For all $w\in V$, with $\text{COLOR}(w) = \text{BLACK}$, 
		\State set $E' \gets E' \cup \{(w,u)\}$ \Comment{Add edges to $u$ from all BLACK nodes.}
		\EndWhile
		\State For all nodes $u$ not considered in the loop above,
		\State set $E' \gets E' \cup \{(w,u): w\in V\}$ \Comment{Add edges to $u$ from all nodes.}
		\Return $G' = (V,E')$.
	\end{algorithmic}
	}
\end{algorithm}

Figure \ref{fig:ZFS_Augmentation_Algo} illustrates an example of edge augmentation due to Algorithm \ref{algo:augmention_zfs}. The derived set in $G$ and $G'$ is same.
\begin{figure}[htb]
    \centering
    \begin{subfigure}[b]{0.239\textwidth}
	\centering
	\includegraphics[scale=0.22]{Main_example.png}
	\caption{$G$}
	\end{subfigure}
	\begin{subfigure}[b]{0.22\textwidth}
	\centering
\includegraphics[scale=0.22]{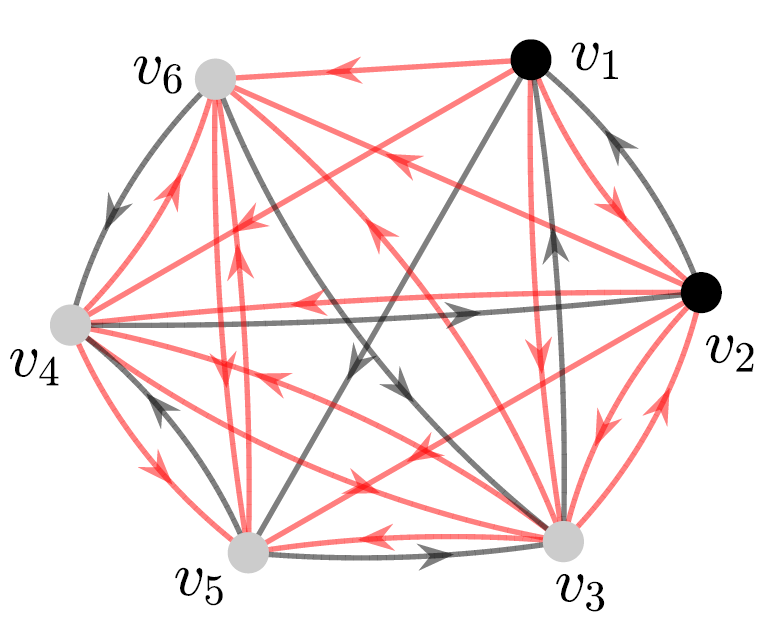}
	\caption{$G'$}
	\end{subfigure}
	\caption{Example of edge augmentation due to Algorithm \ref{algo:augmention_zfs}}.
    \label{fig:ZFS_Augmentation_Algo}
\end{figure}
Next, we show that Algorithm~\ref{algo:augmention_zfs} is optimal and adds the maximum number of edges while preserving the size of the derived set returned by the ZF process. 

\begin{prop}
For a directed graph $G=(V,E)$ with a leader set $V_\ell\subseteq V$, let $G'=(V,E')$ be a graph returned by Algorithm~\ref{algo:augmention_zfs}. Then, $\text{dset}(G,V_\ell)  = \text{dset}(G',V_\ell)$.
\end{prop}
\begin{proof}
Let $\Delta$ be the derived set of $G$ with a leader set $V_\ell$. In $G'$, all leader nodes in $V_\ell$ are colored BLACK due to the initial condition of the ZF process. Consider an arbitrary iteration in the ZF process, where a BLACK node $u$ colors its {\em only} WHITE in-neighbor $v$ BLACK. Algorithm~\ref{algo:augmention_zfs} adds edges in $G'$ from (currently) BLACK nodes to $u$. Since no edge from a currently WHITE node to $u$ is added, $v$ must be the only WHITE in-neighbor of $u$ in $G'$ as well. Thus, the ZF process proceeds by assigning the BLACK color to node $v$, which is the only WHITE in-neighbor of the BLACK node $u$ in $G'$. This holds for every iteration in the ZF process. Thus, the ZF process in $G$ and $G'$ will change the colors of nodes exactly the same way. 
\end{proof}

We can count the number of edges in the directed graph $G'= (V,E')$ returned by Algorithm~\ref{algo:augmention_zfs} in terms of $|V|$, $|V_\ell|$ and the size of the derived set. 

\begin{prop}
\label{prop:graph_size_zfs}
For a graph $G'=(V,E')$ returned by Algorithm~\ref{algo:augmention_zfs} with a leader set $V_\ell$ and derived set $\Delta$, 
\begin{equation*}
\begin{split}
|E'| & \ge \frac{|\Delta| (|\Delta|+1)}{2} - \frac{m (m+1)}{2} + (m+n-|\Delta|)n - n,
\end{split}
\end{equation*}
where $n = |V|$ and $m = |V_\ell|$.
\end{prop}
\begin{proof}
In each iteration of the WHILE loop in Algorithm~\ref{algo:augmention_zfs}, edges from all BLACK nodes to a fixed node are added. There are $|V_\ell| = m$ BLACK nodes when the WHILE loop starts, and this number increases by one in each iteration. Thus, we add $m, m +1,\ldots,|\Delta|-1$ edges in the WHILE loop. Outside the loop, we add $n-1$ incoming edges for each of the remaining $n-|\Delta|+m$ nodes. Therefore,

{\small
\begin{equation*}
\begin{split}
|E'| \ge & \; (n-1) (n-|\Delta|+ m) + \sum_{i=m}^{|\Delta|-1} i\\
= & \;\frac{|\Delta| (|\Delta|+1)}{2} - \frac{m (m+1)}{2} + (m+n-|\Delta|)n - n.
\end{split}
\end{equation*}
}
\end{proof}

\begin{theorem}
\textcolor{blue}{Let $G=(V,E)$ be a directed graph with $|V|=n$ nodes, leader set $V_\ell$ and derived set $\Delta = \text{dset}(G,V_\ell)$, then Algorithm~\ref{algo:augmention_zfs} returns a graph $G^*=(V,E^*)$ where $E^*\supseteq E$ and $|E^*|$ is maximum while preserving the size of the derived set $\Delta$.}
 Moreover, the number of edges in the optimal graph is
\begin{equation*}
 \frac{|\Delta| (|\Delta|+1)}{2} - \frac{m (m+1)}{2} + (m+n-|\Delta|)n - n,
\end{equation*}
where $m = |V_\ell|.$
\end{theorem}

\begin{proof}
\textcolor{blue}{
Let $G^*=(V,E^*)$ be a graph satisfying the conditions stated in the theorem, and $G'=(V,E')$ be a graph returned by Algorithm~\ref{algo:augmention_zfs}.} We will count the number of edges in $G^*$ and show that $|E^*|$ is upper bound by the expression in Proposition~\ref{prop:graph_size_zfs}. Clearly $|E'|$ can not be larger than $|E^*|$, we will get the desired result.

Since $G^*$ preserves the size of the derived set, we should be able to run $|\Delta|- |V_\ell|$ iterations of the ZF process in some arbitrary order. When the ZF process starts, there are $ m$ BLACK nodes and $(n-m)$ WHITE nodes. At this point, there must exist a BLACK node $u$ which has only one in-coming edge from a WHITE in-neighbor. Therefore, at least $n-m-1$ edges of the form $(v,u)$, where $v$ is a WHITE node, are missing from $G^*$. In each iteration, the number of WHITE nodes decreases by exactly one. This means that in the second iteration, at least $(n-m-2)$ edges are missing, and these edges are distinct from previously counted edges because none of these involve the node $u$. Similarly, in iteration $i$, there are at least $(n-m-i)$ distinct edges missing. We can upper bound the number of edges in $G^*$ by subtarting the minimum number of missing edges in the graph from the maximum possible $n^2-n$ edges. Thus,

{\small
 \begin{equation*}
 \begin{split}
 |E^*| \le & \; n^2 - n - \sum_{i=1}^{|\Delta|-m} (n-m-i)
 =   \; n^2 - n - \sum_{ i = n-m- (|\Delta|-m)}^{n-m-1} i\\
 = & \; n^2 - n - \frac{(n-m-1)(n-m)}{2}
 +  \frac{(n-|\Delta|)(n-|\Delta|-1)}{2}.
 \end{split}
 \end{equation*}
}
Using Proposition \ref{prop:graph_size_zfs}, we get $|E^*|\le |E'|$. However, $|E^*|$ is optimal, which means $|E^*|\ge |E'|$. Thus, we deduce that $|E^*|=|E'|$ and conclude the desired statement.
\end{proof}
\textcolor{blue}{The time complexity of Algorithm \ref{algo:augmention_zfs} is $O(n^2)$ where $n$ is the number of nodes in the graph. Note that the main task in the algorithm is the addition of edges to the edge set $E'$. If the edges are kept in an adjacency matrix, each addition takes a constant amount of time. We do not add an edge twice to this set, so the time complexity is bounded by the maximum number of possible edges in a directed graph. We also observe that the time complexity is:
  $$
  \Omega(|E'|) = \Omega( {|\Delta|\choose{2}} + (m+n-|\Delta|)n ).
  $$
When the size of derived set is at least $n/2$, the first term is $\Omega(n^2)$, and when it is less than $n/2$, the second term is $\Omega(n^2)$. Thus, the time complexity of Algorithm \ref{algo:augmention_zfs} is $\Theta(n^2)$.} 

\begin{remark} An important observation here is that the number \textcolor{blue}{of edges that one can add to a directed graph} while preserving the derived set is independent of the topology of the given graph. 
Note, however, that \textcolor{blue}{the size of the derived set} in an arbitrary graph is not independent of the topology.
\end{remark}
\section{Adding Edges through Distance-based Bound}
\label{section:PMI_Story}
In this section, first, we review a tight lower bound on the dimension SSCS based on the distances of nodes to leaders in a graph \cite{yazicioglu2016graph}. Second, we present a method to add edges while preserving distances between specific node pairs, which also preserves the bound on the dimension of SSCS. The distance-based bound on the dimension of SSCS is typically better than the ZF-based bound, especially when the network is not strong structurally controllable \cite{YasinCDC2020}. 
\subsection{Distance-based Bound for SSC}
\label{subsection:PMI_bound}
Given a network with $m$ leaders $V_\ell = \{\ell_1,\cdots,\ell_m\}$, we define the \emph{distance-to-leaders} (DL) vector of each $v_i\in V$ as
\begin{equation*}
\label{eq:DLvector}
D_i = \left[
\begin{array}{ccccc}
d(v_i, \ell_1) & d(v_i,\ell_2)  & \cdots & d(v_i,\ell_m)
\end{array}
\right]^T \in \mathbb{Z}^m.
\end{equation*}
The $j^{th}$ component of $D_i$, denoted by $[D_i]_j$, is equal to the distance of $v_i$ to $\ell_j$. Next, we provide the definition of  \emph{pseudo-monotonically increasing} sequences of DL vectors.

\begin{definition}{(\emph{Pseudo-monotonically Increasing (PMI) Sequence)}}
\label{def:PMI}
A sequence of distance-to-leaders vectors $\calD$ is PMI if for every $i^{th}$ vector in the sequence, denoted by $\calD_i$, there exists some $\pi(i)\in\{1,2,\cdots,m\}$ such that
\begin{equation}
\label{eq:PMIcondition}
[\calD_i]_{\pi(i)} < [\calD_{j}]_{\pi(i)}, \;\;\forall j > i.
\end{equation}
\textcolor{blue}{In other words, \eqref{eq:PMIcondition} needs to be satisfied for all the subsequent distance-to-leader vectors $\mathcal{D}_j$ appearing after $\mathcal{D}_i$ in the sequence.} We say that $\calD_i$ satisfies the \emph{PMI property} at coordinate $\pi(i)$ whenever $[\calD_i]_{\pi(i)} < [\calD_{j}]_{\pi(i)},\;\forall j>i$.
\end{definition}

An example of DL vectors is illustrated in Fig. \ref{fig:PMI}. A PMI sequence of length five can be constructed as


{\small
\begin{equation}
\label{eq:PMIexample}
\calD =
\left[
\left[
\begin{array}{c}
\mathbf{0}\\  3
\end{array}
\right]
,
\left[
\begin{array}{c}
1 \\ \mathbf{0} 
\end{array}
\right]
,
\left[
\begin{array}{c}
\mathbf{1} \\ 4
\end{array}
\right]
,
\left[
\begin{array}{c}
2\\ \mathbf{1}
\end{array}
\right]
,
\left[
\begin{array}{c}
\mathbf{2} \\ 2
\end{array}
\right]
\right].
\end{equation}
}
Indices of bold values in \eqref{eq:PMIexample} are the coordinates, $\pi(i)$, at which the corresponding distance-to-leaders vectors are satisfying the PMI property.

\begin{figure}[htb]
\centering
\includegraphics[scale=0.25]{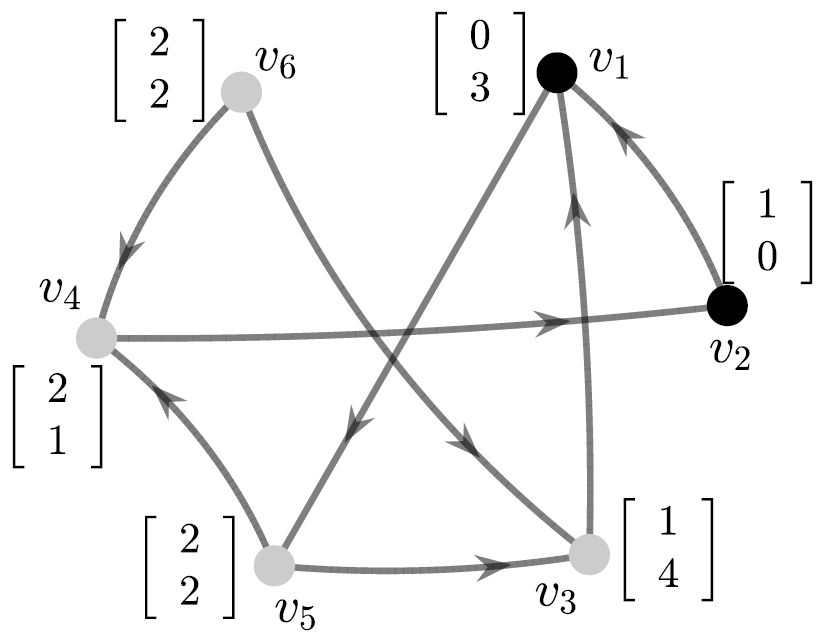}
\caption{A network with two leaders, $V_\ell = \{v_1,v_2\}$, and the corresponding distance-to-leaders (DL) vectors.
}
\label{fig:PMI}
\end{figure}
The longest PMI sequence of DL vectors is related to the dimension of SSCS as stated in the following result.
\begin{theorem} \cite{yazicioglu2016graph}
\label{thm:PMI}
Consider any network $G=(V,E)$ with the leaders $V_\ell \subseteq V$. Let $\delta(G,V_\ell)$ be the length of the longest PMI sequence of distance-to-leaders vectors with at least one finite entry.  Then, $\delta(G,V_\ell) \le \gamma(G,V_\ell)$.
\end{theorem}

\begin{remark}
While the bound in Theorem \ref{thm:PMI} was presented for connected undirected graphs in \cite[Theorem~3.2]{yazicioglu2016graph}, it also holds for any choice of leaders on strongly connected directed graphs as shown in \cite[Remark 3.1]{yazicioglu2016graph}. Such connectivity properties already ensure that all DL vectors have only finite entries. The bound can be extended easily to directed graphs without strong connectivity by excluding the DL vectors with all $\infty$ entries. 
\end{remark}


\subsection{Adding Edges While Preserving Node Distances}
Let $G = (V,E)$ be a graph with a leader set $V_\ell$, $\calD$ be a PMI sequence of length $\delta$, and $\tilde{V}\subseteq V$ be the set of nodes whose DL vectors are included in $\calD$. If we add edges in $G$ to obtain a new \textcolor{blue}{graph} $G' = (V,E')$ such that the DL vectors of nodes in $\tilde{V}$ remain the same in $G'$, then $\calD$ will also be a PMI sequence of $G'$ and $\delta\le\gamma(G',V_\ell)$. Therefore, one approach to augment edges in a graph while preserving a bound on the dimension of SSCS is to ensure that the distances from a certain set of nodes to leaders do not change due to edge additions. In this direction, we first need to study the maximal edge augmentation in a graph while preserving the distance from a given node $a$ to another node $b$. 

\begin{definition}{(\emph{Distance Preserving Edge Augmentation (DPEA) Problem)}} 
\label{def:DPEA}
Given a directed graph $G=(V,E)$ and nodes $a,b\in V$ such that $d_G(a,b)= k$, find a graph $G'(V,E')$ with the (same) node set $V$ and an edge set $E'\supseteq E$ such that $d_{G'}(a,b) = k$ and $|E'|$ is maximized. 
\end{definition}

Next, we characterize optimal solutions of the DPEA problem for a given node pair $(a,b)$ in $G$. We show that the optimal solution belongs to a special class of graphs, which is obtained by the union of \emph{clique chains} and \emph{modified clique chains} described below.

\begin{definition}{(\emph{Directed clique chain)}}
\label{def:CliqueChains}
A directed graph $\calC_k=(V,E)$ is a directed clique chain if the node set $V$ can be partitioned into sets $V_0,V_1,V_2,\ldots,V_k$, such that there is an edge from every node in $V_i$ to every node in $V_{i-1}\cup  V_{i} \cup V_{i+1}$ for all $1\le i\le k-1$. Moreover, nodes in each of $V_0$ and $V_k$ induce cliques.\footnote{In a clique, there is an edge between every pair of nodes.}
\end{definition}

\begin{definition}{(\emph{Directed modified clique chain)}}
\label{def:Other_graph}
A directed graph $\mathcal{M}_k=(V,E)$ is a directed modified clique chain if the node set $V$ can be partitioned into sets $V_0,V_1,V_2,\ldots,V_k$, such that there is an edge from every node in $V_i$ to every node in $V_{j}$ for all $j\le i$.
\end{definition}

Examples of directed clique chain and directed modified clique chain are shown in Figure \ref{fig:CCOther}. 

\begin{figure}[htb]
    \centering
    \begin{subfigure}[b]{0.22\textwidth}
	\centering
	\includegraphics[scale=0.61]{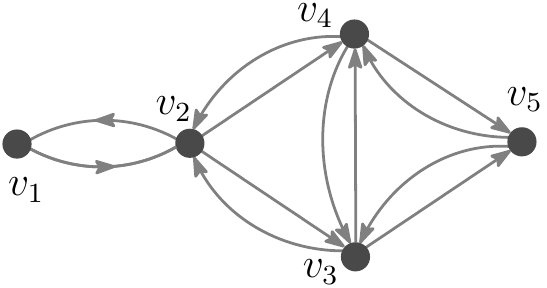}
	\caption{$\calC_3$}
	\end{subfigure}
	\begin{subfigure}[b]{0.22\textwidth}
	\centering
	\includegraphics[scale=0.6]{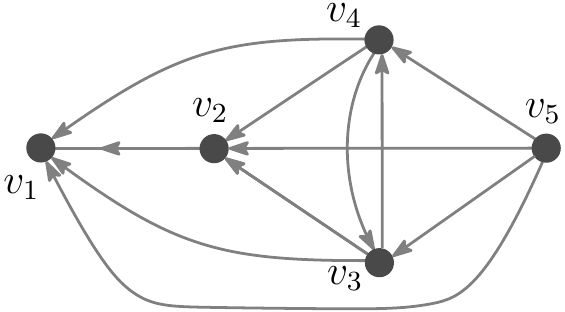}
	\caption{$\calM_3$}
	\end{subfigure}
    \caption{(a) Directed clique chain and (b) directed modified clique chain. The node set $V$ is partitioned into four subsets: $V_0 = \{v_1\}$, $V_1 = \{v_2\}$, $V_2 = \{v_3, v_4\}$ and $V_3 = \{v_5\}$.}
    \label{fig:CCOther}
\end{figure}

\begin{theorem}
	\label{thm:directed_augmented}
	Let $G=(V,E)$ be a directed graph and let $a,b$ be two fixed vertices in $G$ with $d_G(a,b) = k$. Then, an augmented graph $G'=(V,E'),E'\supseteq E$, that preserves the distance from $a$ to $b$, and contains the maximum number of edges is a union of a directed clique chain $\calC_k$  and a directed modified clique chain $\mathcal{M}_k$ for some partition $V_0=\{a\},V_1,V_2,\ldots,V_k=\{b\}$ of the node set $V$. 
\end{theorem}
\begin{proof} Let $d_G(a,b) = k$. Then, it is clear that for all vertices $v\in V$, $d_{G'}(a,v)\le k$ because otherwise we can add an edge from an arbitrary node at distance $\le k-1$ to $v$.
Also, an arbitrary vertex $v$ in $G'$ at distance $i$ from $a$, we may assume that $d_{G'}(v,b) = k - i$. Clearly, $d_{G'}(v,b)$ can not be less than $k - i$, and if it is more, we can add an edge from $v$ to a vertex $w$ where $d_{G'}(w,b) = k - i -1$. Thus, every vertex in $G'$ lies on a shortest path from $a$ to $b$.
Let $u,v$ be two vertices such that $d_{G'}(a,u) = i$, and $d_{G'}(a,v) = j$. We have the following cases:
\begin{enumerate}
	\item $i=j$ or $i=j-1$ : an edge from $u$ to $v$ doesn't change the distance from $a$ to $b$. Thus, we may assume that all such edges exist in $G'$.
	\item $i < j-1$: since both of $u,v$ lie on a shortest path from $a$ to $b$, we know that $d_{G'}(v,b) = k-j$. So an edge $u,v$ will create a path from $a$ to $b$ of distance $\le j-2 +1 + k -j = k-1$. This is a contradiction to the fact that $G'$ preserves $a,b$ distance. Therefore, $G'$ doesn't contain any edge of the form $u,v$.
	\item $i>j$: in this case, an edge from $u$ to $v$ doesn't create any new shortest paths from $a$ to $b$. Thus, we may assume that all such edges exist in $G'$.
\end{enumerate}
Let us define $V_0,V_1,\ldots,V_k$ a partition of $V$ where $V_i = \{v\in V: d_{G'}(a,v) = i \} $. It is clear that with this partition, the graph $G$ is a union of a directed clique chain and a directed modified clique chain. This provides a complete characterization of edges in $G'$ and completes the proof.
\end{proof}



From Theorem~\ref{thm:directed_augmented}, we obtain a simple way to greedily construct a maximally dense graph that preserves the distance from node $a$ to node $b$ in a given graph $G$ as follows:

\begin{algorithm}
\caption{Distance Preserving Edge Augmentation (DPEA)}
\label{algo:DPEA}
{\small
\begin{algorithmic}[1]
\State \textbf{Given} $G = (V,E)$, $a,b\in V$. 
\State \textbf{Initialize} $E' \gets E$, $G' \gets (V,E')$
\State \textbf{While} $\exists$ $(u,v)\in V\times V$ with $d_{G'}(a,u) \ge d_{G'}(a,v) -1$
\State 
$
E' \gets E' \cup\{(u,v)\}
$, $G' \gets (V,E')$.
\State  \textbf{End While}

\State \textbf{Return} $G' = (V,E')$
\end{algorithmic}
}
\end{algorithm}
An example of the DPEA is shown in Figure \ref{fig:DPEA} in which edges are augmented to preserve the distance from node $a=v_1$ to node $b = v_5$.

\begin{figure}[htb]
    \centering
    \begin{subfigure}[b]{0.239\textwidth}
	\centering
	\includegraphics[scale=0.64]{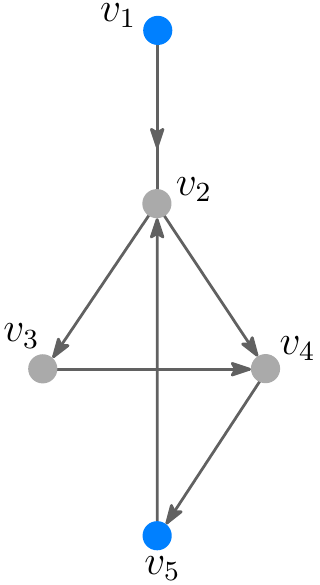}
	\caption{$G$}
	\end{subfigure}
	\begin{subfigure}[b]{0.239\textwidth}
	\centering
	\includegraphics[scale=0.64]{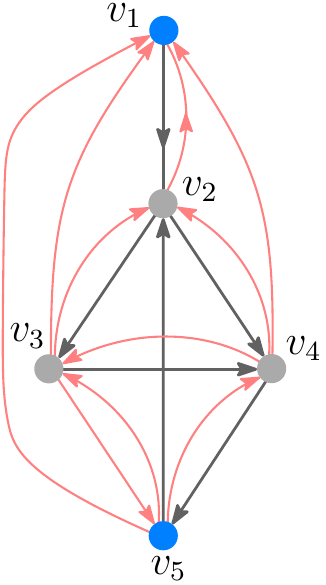}
	\caption{$G'$}
	\end{subfigure}
    \caption{Distance preserving edge augmentation for the node pair $(v_1,v_5)$. In (b), red edges are the ones that are added as a result of Algorithm \ref{algo:DPEA}.}
    \label{fig:DPEA}
\end{figure}
\textcolor{blue}{We note that the time complexity of Algorithm \ref{algo:DPEA} is $O(n^2\log n)$ because we can implement this using runs of Dijkstra's shortest path algorithm}. We can use DPEA to add edges in a graph $G(V,E)$ with a leader set $V_\ell$ while preserving a bound on the dimension of SSCS. Let $\calD$ be a PMI sequence of length $\delta$ and $\tilde{V}\subseteq V$ be the set of nodes whose DL vectors are included in $\calD$. By solving the DPEA problem for the node pair $(v,l)$, where $v\in\tilde{V}$ and $l\in V_\ell$, we can obtain edges, say $E_{v,l}$, whose addition to the graph will preserve the distance from node $v$ to leader $l$. By solving DPEA problem for all node pairs $(v,l)$, where $v\in\tilde{V}$ and $l\in V_\ell$, we can obtain edges that are common in all solutions, that is, $E_{comm}= \cap_{v\in\tilde{V},l\in V_\ell}E_{v,l}$. By adding these common edges in the given graph $G$, we obtain a new graph $G' = (V,E\cup E_{comm})$ such that $d_G(v,l) = d_{G'}(v,l)$, $\forall v\in \tilde{V}$ and $\forall l \in V_\ell$. Consequently, $\calD$ will also be PMI sequence of $G'$, which means that the dimension of SSCS in the augmented graph $G'$ will also be at least $\delta$. Thus, by solving multiple instances of DPEA problem, we can determine edges whose addition to the graph will preserve the distance-based bound on the dimension of SSCS.  

\textcolor{blue}{
Since the problem of finding an optimal set of edges while preserving the distance-based bound is computationally challenging, in the next subsection, we present and analyze a \emph{randomized} edge augmentation algorithm that is simple to understand, easy to implement, and offers good numerical results in our experiments.}
\subsection{Randomized Edge Augmentation Algorithm Preserving the Distance-based Bound}
First, we discuss that for a given $G$, if $\calD$ is a PMI sequence of length $\delta$ containing DL vectors of nodes $\tilde{V}\subseteq V$, then we can add edges in $G$ to obtain an augmented graph $G'$ with the same leader set $V_\ell$. Moreover $G'$ will also have a PMI sequence of length $\delta$ even if $d_G(v,l) \ne d_{G'}(v,l)$, $\forall v\in \tilde{V}, \forall l \in V_\ell$. In particular, for every node pair $(v,l)$, where $v\in\tilde{V}$, $l\in V_\ell$, there exists an integer $\epsilon_{v,l}\in \mathbb{Z}$ such that if $\epsilon_{v,l} < d_{G'}(v,l)\le d_G(v,l)$, then $G'$ will have a PMI sequence consisting of DL vectors of nodes in $\tilde{V}$ and having a length $\delta$. 
We will then provide a randomized edge augmentation algorithm satisfying conditions to preserve the distance-based bound on the dimension of SSCS. Finally, we will analyze the performance of the algorithm.

Let $\calD_i$ be the $i^{th}$ vector in the PMI sequence $\calD$ of the given graph $G$. Then by the definition of PMI sequence, for each element $[\calD_i]_j$, there is an integer, say $\epsilon_{i,j}$, such that $[\calD_i]_j > \epsilon_{i,j}$. We denote the maximum possible value of $\epsilon_{i,j}$ by $\epsilon^\ast_{i,j}$ and define $\epsilon^\ast_i := \left[\begin{array}{rrrr}\epsilon_{i,1}^\ast & \epsilon_{i,2}^\ast &\cdots & \epsilon_{i,|V_\ell|}^\ast\end{array}\right]^T$. For instance, consider the PMI sequence in \eqref{eq:PMIexample}. The vectors $\calD_i$ in the sequence and the corresponding $\epsilon^\ast_i$ are given below.

{\small
\begin{equation*}
\begin{split}
& \calD_1 = \left[\begin{array}{c} 0 \\ 3
\end{array}\right], \;\epsilon^\ast_1 = \left[\begin{array}{c}
-1 \\ -1 \end{array}\right];\;\;
\calD_2 = \left[\begin{array}{c} 1 \\ 0
\end{array}\right], \;\epsilon^\ast_2 = \left[\begin{array}{c}
0 \\ -1 \end{array}\right];\\
& \calD_3 = \left[\begin{array}{c} 1 \\ 4
\end{array}\right], \;\epsilon^\ast_3 = \left[\begin{array}{c}
0 \\ 0 \end{array}\right];\;\;
\calD_4 = \left[\begin{array}{c} 2 \\ 1
\end{array}\right], \;\epsilon^\ast_4 = \left[\begin{array}{c}
1 \\ 0 \end{array}\right];\\
& \calD_5 = \left[\begin{array}{c} 2 \\ 2
\end{array}\right], \;\epsilon^\ast_5 = \left[\begin{array}{c}
1 \\ 1 \end{array}\right];\\
\end{split}
\end{equation*}
}

We can think of $\epsilon^\ast_{i,j}$ as a strict lower bound on $[\calD_i]_j$. From a given PMI sequence of length $\delta$, we can always obtain $\epsilon^\ast_{i}$ for all $i\in\{1,2,\cdots,\delta\}$. 

\begin{observation}
In a PMI sequence $\calD$, if we replace  $[\calD_i]_j$ by some integer $x$, where $\epsilon^\ast_{i,j}< x \le [\calD_i]_j$, then the resulting sequence will still be a PMI sequence.
\end{observation}
For instance, consider the sequence $\bar{\calD}$ below, which is obtained from $\calD$ in \eqref{eq:PMIexample} by replacing $[\calD_1]_2 = 3$ and $[\calD_3]_2 = 4$ by $[\bar{\calD}_1]_2 = 1$ and $[\bar{\calD}_3]_2 = 1$, respectively. Since $\epsilon^\ast_{1,2} = -1$ and  $\epsilon^\ast_{3,2} = 0$, the resulting sequence $\bar{\calD}$ is a PMI sequence.

{\small
\begin{equation}
\label{eq:PMI_new_seq}
\bar{\calD} = \left[
\begin{array}{ccccc}
\left[\begin{array}{c} 0 \\ 1\end{array} \right],
\left[\begin{array}{c} 1 \\ 0\end{array} \right],
\left[\begin{array}{c} 1 \\ 1\end{array} \right],
\left[\begin{array}{c} 2 \\ 1\end{array} \right],
\left[\begin{array}{c} 2 \\ 2\end{array} \right]
\end{array}
\right].
\end{equation}
}

Next, we present a randomized algorithm to add edges in a graph $G$ with a leader set $V_\ell$. $G$ has a PMI sequence $\calD$ of length $\delta$ consisting of DL vectors of nodes $\tilde{V}\subseteq V$. Moreover, let $S =\left[ \begin{array}{llllll}s_1&s_2&\cdots &s_\delta\end{array}\right]$, where $s_i$ is the index of the node whose DL vector is the $i^{th}$ element (vector) in $\calD$. For instance, $\calD$ in \eqref{eq:PMIexample} is a PMI sequence of the graph in Figure \ref{fig:PMI}. The corresponding $\tilde{V}$ is $\{v_1,v_2,\cdots,v_5\}$ and $S$ is $\left[ \begin{array}{llllll} 1 & 2 & 3 & 4 & 5 \end{array}\right]$. For every $i^{th}$ vector in the sequence, the algorithm first computes the corresponding vector $\epsilon_i^\ast$, which basically provides minimum distances that need to be maintained from node $v_{s_i}$ to all leaders during the edge augmentation. The algorithm then selects a missing edge randomly and augments it to the graph if its addition does not violate distance conditions (line 8 in Algorithm \ref{algo:RA}), otherwise discards it. This process is repeated until all missing edges from the original graph are either added to the graph, or discarded. The details are outlined in Algorithm \ref{algo:RA}.
\begin{algorithm}
\caption{Randomized Algorithm for the Distance-based Edge Augmentation}
\label{algo:RA}
{\small
\begin{algorithmic}[1]
\State \textbf{Given} $G = (V,E)$, $V_\ell=\{\ell_1,\cdots,\ell_m\}$, $\mathcal{D}$, $S$.
\State \textbf{Initialize} $E' \gets E$
\State Compute $\epsilon^\ast_{i,j}$ for each element $[\calD_i]_j$ in $\calD$.
\State Compute $E^c$ (set of all missing edges).
\State \textbf{While} $E^c\ne \emptyset$
\State Randomly select $e\in E^c$, and obtain $H=(V,E'\cup\{e\})$.
\State Compute $d_H(v_{s_i},\ell_j)$ for all $j\in \{1,\cdots,m\}$ and for all $i\in \{1,\cdots,|S|\}$.  
\State If $(\epsilon^\ast_{i,j} < d_H(v_{s_i},\ell_j) \le [\calD_i]_j$ for all $j\in \{1,\cdots,m\}$ and for all $i\in \{1,\cdots,|S|\}$, then
$
E' \gets E' \cup\{e\}.
$
\State Update $E^c\gets E^c\setminus\{e\}$.
\State  \textbf{End While}
\State \textbf{Return} $E'$
\end{algorithmic}
}
\end{algorithm}

Figure \ref{fig:PMI_Augmentation_Algo}(a) illustrates an example of edge augmentation due to Algorithm \ref{algo:RA} with the graph $G$ in Figure \ref{fig:PMI} and a PMI sequence of length $4$ as inputs. The red edges indicate the augmented edges. We observe that the graph $G'$ in Figure~\ref{fig:PMI_Augmentation_Algo}(a) contains a total of 29 edges compared to the graph in Figure~\ref{fig:ZFS_Augmentation_Algo}(b), which contains a total of 25 edges after the ZF-based edge augmentation.

\begin{figure}[htb]
    \centering
    \begin{subfigure}[b]{0.239\textwidth}
	\centering
	\includegraphics[scale=0.21]{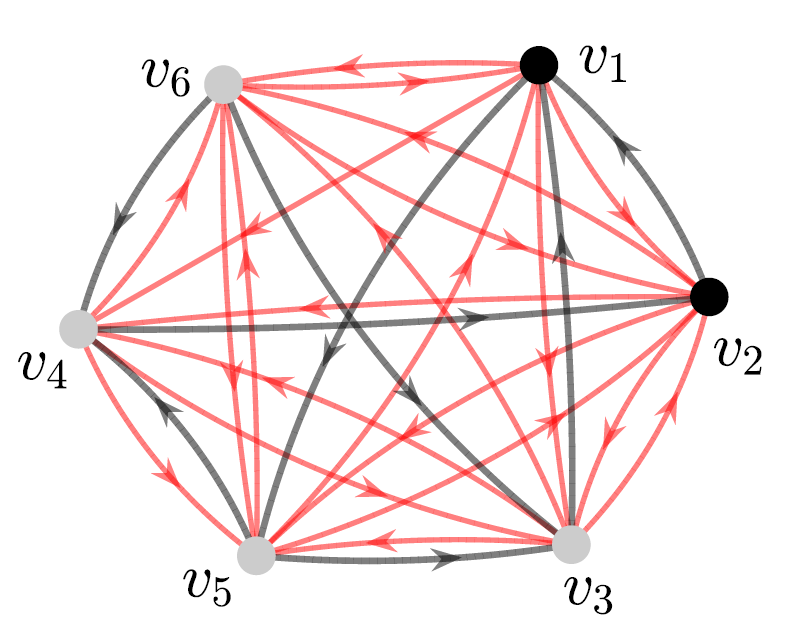}
	\caption{$G'$}
	\end{subfigure}
	\begin{subfigure}[b]{0.24\textwidth}
	\centering
	\includegraphics[scale=0.21]{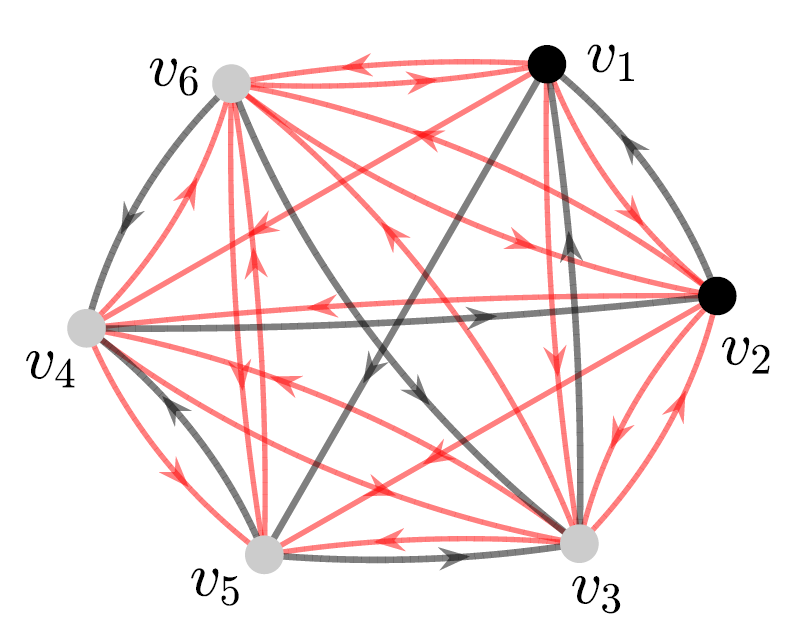}
	\caption{$G''$}
	\end{subfigure}
    \caption{Distance-based edge augmentation due to Algorithm~\ref{algo:RA}. (a) $G'$ is the augmented graph when PMI sequence of length 4 consisting of DL vectors of nodes $v_1,\cdots, v_4$ is the input to the algorithm. (b) $G''$ is the augmented graph with a PMI sequence of length 5 (given in \eqref{eq:PMIexample}) as the input.}
    \label{fig:PMI_Augmentation_Algo}
\end{figure}

\begin{remark}
We note that the distance-based bound on the dimension of SSCS is typically better than the ZF-based bound, especially when the graph is not SSC \cite{YasinCDC2020}. Thus, edge augmentation using Algorithm \ref{algo:RA} allows to add edges while preserving a better bound on the dimension of SSCS. For instance, the distance-based bound on the dimension of SSCS in $G$ is 5, which is better than the ZF-based bound. We see that using Algorithm \ref{algo:RA}, we can have a total of 27 edges in the augmented graph $G''$ (as illustrated in Figure~\ref{fig:PMI_Augmentation_Algo}(b)) while maintaining the dimension of SSCS to be 5, which is not possible by the edge augmentation using the ZF-based bound.
\end{remark}


Next, we analyze the performance of Algorithm \ref{algo:RA} by the following result.
\begin{prop}
If we repeat Algorithm~\ref{algo:RA} a constant $c$ number of times and then return the best graph among these iterations, then the returned graph is an $\alpha$--approximation with probability at least $
1 - e^{- c \left( \frac{OPT}{\beta}\right)^{\alpha\times OPT}}$,
where $OPT$ is the optimal number of edges and \textcolor{blue}{ $\beta$ is the number of edges that can each be added without changing the PMI.}
\end{prop}

We omit the proof of this proposition, because it follows the same arguments given in \cite[Proposition 4.2]{AbbasACC2020} for a similar result on the undirected graphs. \textcolor{blue}{Moreover, the time complexity of Algorithm \ref{algo:RA} is $O(|E|\times|E^c|\times m\log n)$, where $|E^c|$ is the number of edges not in the input graph. The term $O(|E|\times m\log n)$ is the cost of running Dijkstra's shortest path algorithm for all leader nodes and we do it for all of the missing edges $E^c$ in the randomized algorithm}.

\section{Numerical Evaluation}
\label{sec:Num_Eval}
We illustrate and compare the ZF-based and the distance-based edge augmentation algorithms on random directed networks with $N = 100$ nodes in which edge $(i,j)$ exists with probability $p=0.075$, $\forall i\ne j$. Each point in plots in Figure~\ref{fig:main_100} is an average of 30 randomly generated instances.

In Figure \ref{fig:main_100}(a), we plot ZF-based and distance-based bounds on the dimension of SSCS as a function of number of leaders, which are chosen randomly. The distance-based bound is better than the ZF-based, especially for a smaller number of leaders (as discussed in \cite{YasinCDC2020}). Figure \ref{fig:main_100}(b) plots the number of edges added in graphs as a function of number of (randomly selected) leaders. We note that Algorithms \ref{algo:augmention_zfs} and \ref{algo:RA} augment edges while preserving the ZF-based and distance-based bounds on the dimension of SSCS, respectively. For the same number of leaders, the number of edges augmented by Algorithm \ref{algo:augmention_zfs} is greater than the Algorithm \ref{algo:RA} because the controllability bound preserved by the ZF-based augmentation (Algorithm \ref{algo:augmention_zfs}) is smaller than the distance-based augmentation. The number of edges in the original graphs are also shown. Figure \ref{fig:main_100}(c) illustrates the result when we augment edges using Algorithms \ref{algo:augmention_zfs} and \ref{algo:RA} while preserving the same controllability bound, which is ZFS-based. 
\begin{figure}
    \centering
    \begin{subfigure}[b]{0.22\textwidth}
	\centering
	\includegraphics[scale=0.27]{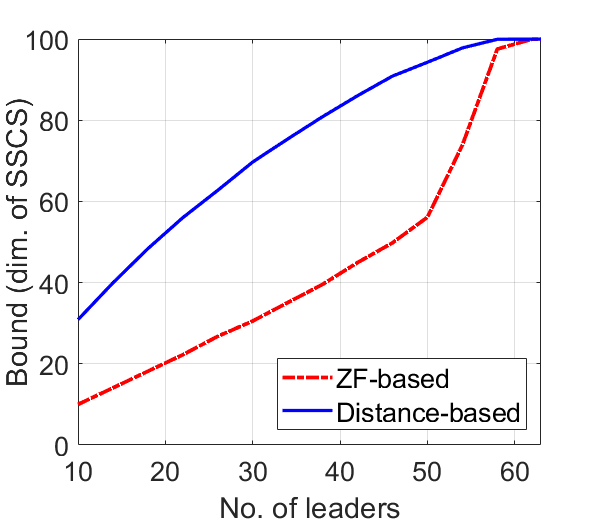}
	\caption{}
	\end{subfigure}
	 \begin{subfigure}[b]{0.22\textwidth}
	\centering
	\includegraphics[scale=0.2225]{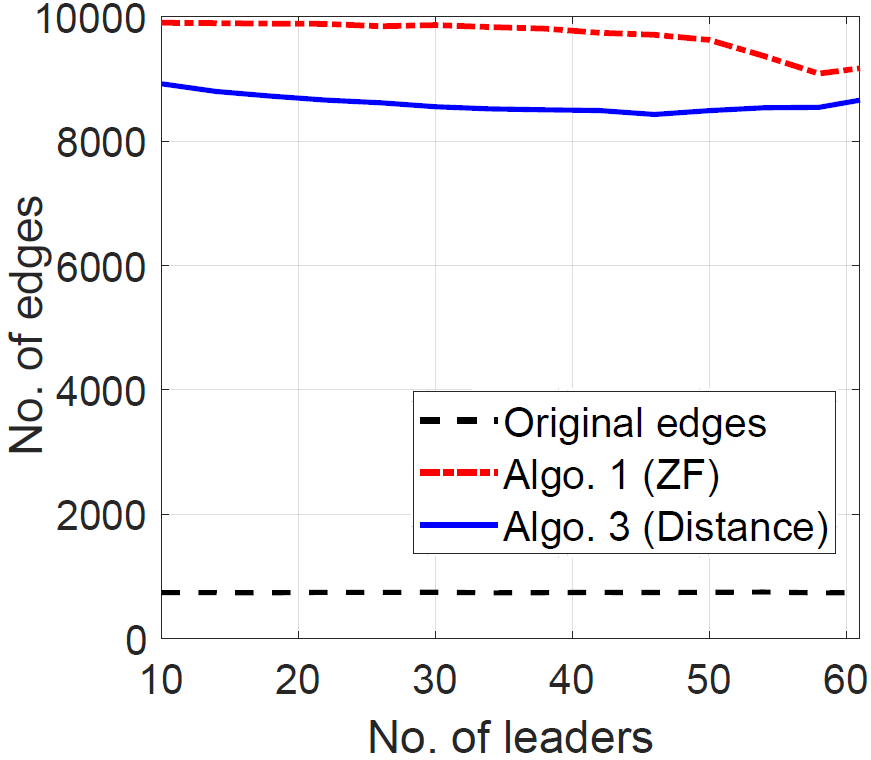}
	\caption{}
	\end{subfigure}
	\begin{subfigure}[b]{0.22\textwidth}
	\centering
	\includegraphics[scale=0.2775]{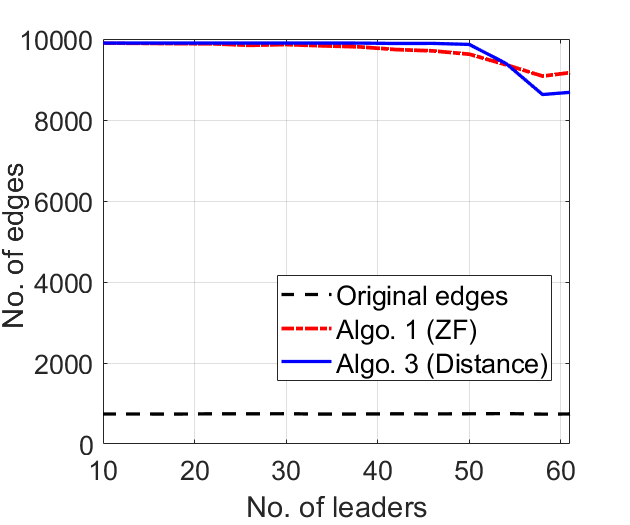}
	\caption{}
	\end{subfigure}
    \caption{(a) Lower bounds on the dimension of SSCS. (b)~Number of edges added by ZF-based and distance-based augmentation algorithms while preserving their respective bounds. (c) Number of edges added by Algorithms~\ref{algo:augmention_zfs}~and~\ref{algo:RA} while preserving the same (ZF-based) bound.}
    \label{fig:main_100}
\end{figure}
\section{Conclusion}
\label{sec:Con}
Network connectivity and robustness can be improved by augmenting extra links between nodes. However, adding new links may degrade the network's controllability. 
In this paper, we presented edge augmentation algorithms to add the maximum number of edges in a network while preserving the ZF-based bound and the distance-based bounds on the dimension of SSCS. When the bound on the dimension of SSCS to be preserved is smaller, a large number of edges can be augmented. 
\textcolor{blue}{Though we considered networks with Laplacian dynamic \eqref{eq:dynamics}, both the distance-based and the ZF-based methods are applicable to more generalized dynamics in the form $\dot{x}= Px + Bu$ (e.g., \cite{Van2017distance,monshizadeh2014zero}). In particular, the ZF-based bound holds for any such linear dynamics on a network where an edge $(v_i,v_j)$ denotes that the corresponding entry in the  system matrix, $P_{ij}$, is non-zero. In comparison, the distance-based method requires the system matrix, $P$, to be in a class of matrices called the \emph{distance-information-preserving matrices} (as explained in \cite{Van2017distance}), which contain the graph Laplacian as a special case.}
We aim to further explore the relation between network controllability and edge density to co-optimize robustness and controllability in the future.


\bibliographystyle{IEEEtran}
\bibliography{references}
\end{document}